\pdfoutput=1
\documentclass[a4paper]{lipics-v2021}
\usepackage{amssymb,amsmath,amsthm}
\usepackage{todonotes}
\usepackage{hyperref}
\usepackage{apxproof}
\usepackage[ruled,vlined]{algorithm2e}
\hideLIPIcs
\nolinenumbers

\title{Bounds on BDD-Based Bucket Elimination}
\author{Stefan Mengel}{Univ.~Artois, CNRS, Centre de Recherche en Informatique de Lens (CRIL)}{}{}{}
\authorrunning{S.~Mengel}
\funding{This work has been partly supported by the PING/ACK project of the French National Agency for Research (ANR-18-CE40-0011). }
\category{Short Paper}

\EventEditors{Meena Mahajan and Friedrich Slivovsky}
\EventNoEds{2}
\EventLongTitle{26th International Conference on Theory and Applications of Satisfiability Testing (SAT 2023)}
\EventShortTitle{SAT 2023}
\EventAcronym{SAT}
\EventYear{2023}
\EventDate{July 4--8, 2023}
\EventLocation{Alghero, Italy}
\EventLogo{}
\SeriesVolume{271}
\ArticleNo{15}

\Copyright{Stefan Mengel}

\newtheoremrep{lemma}{Lemma}

\newcommand{\OBDDsystem}{\text{OBDD}(\land, \exists)}

\begin{document}

\maketitle

\begin{abstract}
We study BDD-based bucket elimination, an approach to satisfiability testing using variable elimination which has seen several practical implementations in the past. We prove that it allows solving the standard pigeonhole principle formulas efficiently, when allowing different orders for variable elimination and BDD-representations, a variant of bucket elimination that was recently introduced. Furthermore, we show that this upper bound is somewhat brittle as for formulas which we get from the pigeonhole principle by restriction, i.e., fixing some of the variables, the same approach with the same variable orders has exponential runtime. We also show that the more common implementation of bucket elimination using the same order for variable elimination and the BDDs has exponential runtime for the pigeonhole principle when using either of the two orders from our upper bound, which suggests that the combination of both is the key to efficiency in the setting.
\end{abstract}

\section{Introduction}

We analyze several aspects of a simple approach to propositional satisfiability called \emph{bucket elimination} based on binary decision diagrams (BDDs)~\cite{Bryant86}. It was originally introduced by Pan and Vardi~\cite{PanV05}, and works, given a CNF $F$, as follows: first translate all clauses of $F$ into BDDs, all having the same variable order. Then, along another variable order, conjoin all BDDs that contain the current variable $x$, eliminate $x$ in the result of the conjoin operation  by existential quantification, add the resulting BDD to the current set of BDDs, and finally delete all BDDs containing $x$. The end result is a BDD representing one of the constants $1$ or $0$, depending on if $F$ is satisfiable or not. The algorithm is often described by putting the BDDs in \emph{buckets} treated in the variable order as in the pseudocode Algorithm~\ref{alg}. It is not hard to see that this approach decides satisfiability of all CNF-formulas correctly. We remark in passing that bucket elimination has also been used as a general approach for reasoning in artificial intelligence~\cite{Dechter99}. In particular, in the context of propositional satisfiability one can implement ordered resolution, also called Davis-Putnam resolution~\cite{DavisP60}, with it, which leads to an algorithm that is similar to what we described above~\cite{RishD00} but uses CNF-formulas to represent intermediate results and not BDDs. In the remainder, we will only focus on bucket elimination that is based on BDDs.

Several SAT-solvers using bucket elimination have been implemented~\cite{PanV05,JussilaSB06,BryantH21}, also motivated by a relation to extended resolution which allows extracting clausal refutations of CNF-formulas efficiently from runs of bucket elimination. 

In this paper, we aim to get a theoretical understanding of the strength of bucket elimination. We first prove that the approach is powerful enough to efficiently solve the well-known pigeonhole principle formulas $\text{PHP}_n$ which are hard for other techniques, in particular resolution~\cite{Haken85}. Our bound confirms recent experimental results for a different encoding that was specifically chosen to make the algorithm efficient~\cite{CodelRHB21,CodelRHB21b}. We here show that also for the standard encoding, there is a choice of variable orders with which bucket elimination can efficiently solve pigeonhole principle formulas.

We then go on showing that the upper bound for $\mathsf{PHP}_n$ is in a sense brittle: one can restrict the formula $\text{PHP}_n$ by assigning some of its variables, resulting in a formula on which bucket elimination with the same variable orders as before takes exponential time. This is surprising since fixing some of the variables reduces the search space and thus should make the problem easier. However, in the case of bucket elimination it has the opposite effect, making the runtime explode. 
This suggests that bucket elimination is not very stable under small variations of the input.

The final part of this paper is motivated by the fact that the pigeonhole principle has been used as a benchmark also  in~\cite{PanV05,BryantH21} where bucket elimination was shown to be practically inefficient. The difference between our result and~\cite{CodelRHB21b} on the one hand and~\cite{PanV05,BryantH21} on the other hand is that the latter, as also the implementation of~\cite{JussilaSB06}, consider the same  variable order for the variable elimination and the order in the BDDs. In contrast, in our result and the current public version of the implementation of~\cite{BryantH21,CodelRHB21,CodelRHB21b}\footnote{\url{https://github.com/rebryant/pgbdd}} two different orders may be chosen. To explore the impact of this change, we consider bucket elimination for $\text{PHP}_n$ where only one of the variable orders we use in our upper bound is used. We show that in both cases the variant that uses only one order has exponential runtime, which shows that to efficiently solve $\text{PHP}_n$ the combination of the two orders is crucial and is more powerful than each of them individually.


Our results can also be seen in the context of BDD-based proof systems, more specifically, they are close to results on the proof system $\OBDDsystem$~\cite{AtseriasKV04} which allows general conjunction and variable elimination without any scheduling restrictions. It was shown in~\cite{ChenZ09} that there are polynomial size refutations of the pigeonhole principle in this system. This also follows from our result, which can be interpreted as working in a restricted fragment of $\OBDDsystem$. We remark also that \cite{BussIKRS21} claims that the proofs in~\cite{ChenZ09} can be implemented in the algorithm of~\cite{PanV05}. However, this seems to be not the case due to the order restrictions of that algorithm which are not respected in the proof. It is however possible that, after rearranging the operations, the proof in~\cite{ChenZ09} could be implemented with two orders, similarly to our result.

\begin{algorithm}[t]
	\SetAlgoLined
	\LinesNumbered
	\SetAlgoNoEnd
	\DontPrintSemicolon
	\KwIn{clauses $C_1, \ldots, C_m$ in variable set $X$, elimination variable order $\pi$}
	\For{$x\in X$}{
		create empty bucket $B_x$
	}
	\For{$i= 1, \ldots, m$} {
		compute a BDD for $C_i$, put it into $B_y$ where $y$ is first variable in order $\pi$ in $C_i$
	}
	\For{$x\in X$ in order $\pi$}{
		compute BDD $D$ by iteratively conjoining all BDDs in $B_x$\;\label{ln:conjoin}
		\If{$D$ is constant $0$-BDD} {\Return 0}
		compute a BDD $D'$ computing $\exists x D$ and put it into $B_y$ where $y$ is first variable in order $\pi$ in $D'$\label{ln:project}
	}
	\Return{1}
	\caption{BDD-based bucket elimination for CNF-formulas}
	\label{alg}
\end{algorithm}

\section{Preliminaries}

We use the usual integer interval notation, e.g.~$[n]:=\{1, \ldots, n\}$ and $[m,n]:= \{m, m+1, \ldots, n-1, n\}$. When speaking of graphs, we mean finite, simple, undirected graphs. We write $G=(A,B,E)$ for a bipartite graph with color classes $A$ and $B$ and edge set $E$. 

We assume that the reader is familiar with the basics of propositional satisfiability, in particular CNF-formulas, see e.g. the introductory chapters of~\cite{handbook21}. Given a CNF-formula $F$ and a partial assignment $a$, we call the \emph{restriction} of $F$ by $a$ the CNF which we get by fixing the variables according to $a$ and simplifying, i.e., we delete all clauses that are satisfied by $a$ and from the other clauses all literals that are falsified.


An \emph{(ordered) binary decision diagram} (short \emph{BDD} or OBDD) is a graph-based representa\-tion of Boolean functions as follows~\cite{Bryant86}: a BDD over a variable set $X$ consists of a directed acyclic graph with one source and two sinks. The sinks are labeled $0$ and $1$, respectively, while all other nodes are labeled by variables from $X$. Every node but the sinks has two out-going edges, called $0$-edge and $1$-edge, respectively. Given an assignment $a$ to $X$, we construct a source-sink path in the BDD starting in the source and iteratively following the $a(x)$-edge to the next node, where $x$ is the label of the current node. Eventually, we end up in a sink whose label is the value computed by the BDD on $a$. This way, the BDD specifies a Boolean value for every assignment to $X$ and thus defines a Boolean function. BDDs are required to be ordered as follows: there is an order $\pi$ on $X$ such that whenever there is an edge from a node labeled by $x$ to a node labeled by $y$, then $x$ appears before $y$ in $\pi$. It follows that on every source-sink path one encounters every variable at most once.

It will sometimes be convenient to reason with complete BDDs which are BDDs in which all source-sink paths contain all variables as labels. The \emph{width} of a complete BDD is defined as the maximal number of nodes that are labeled by the same variable. Clearly, a complete BDD in $n$ variables and of width $w$ has size at most $O(nw)$. Moreover, it is well known that when conjoining two BDDs with the same variable order and width $w_1$ and $w_2$, respectively, the result has the same order and width at most $w_1\cdot w_2$.


We will use the following known lower bound, see e.g.~\cite[Section~6]{AmarilliCMS20}; for the convenience of the reader, we give a self-contained proof in the appendix.
\begin{lemmarep}\label{lem:lowerscov}
	Every BDD computing $\bigwedge_{i\in [n]}x_i\lor y_i$ with a variable order in which every $x_i$ comes before every $y_j$ has at least $2^n$ nodes.
\end{lemmarep}
\begin{proof}
	Consider the partition $(X,Y)$ where $X$ contains all the $x_i$ and $Y$ all the $y_i$. We start by presenting a known connection to so-called rectangle covers. A \emph{rectangle} respecting the partition $(X,Y)$ is a function $f(X,Y)$ that can be written as a conjunction
	\begin{align*}
		f(X,Y):= f_1(X)\land f_2(Y).
	\end{align*}
	All rectangles we consider here will respect $(X,Y)$, so we do not mention it here explicitly in the remainder. We say that an assignment $a$ \emph{lies in the rectangle} $f$, if $f(a)= 1$. A \emph{rectangle cover} of a function $f(X,Y)$ is a sequence $f^1, \ldots, f^s$ such that
	\begin{align}
		f(X,Y) = \bigvee_{i\in [s]} f^i(X,Y)
	\end{align}
	where the $f^i$ are all rectangles. The size of the rectangle cover is defined to be $s$. Rectangles are connected to BDDs due to the following result:
	\begin{claim}\label{clm:rectangles}
	If a function $f$ can be represented by a BDD of size~$s$ with a variable order in which all variables in $X$ appear before those in $Y$, then there is a rectangle cover of $f$ of size~$s$.
	\end{claim}
	\begin{claimproof}
		Let $D$ be a BDD computing $f$.	Let $v_1, \ldots, v_\ell$ the nodes with a label not in $X$ such that there is an edge from a node with label in $X$ to $v_i$. Remember that every assignment $a$ to $(X,Y)$ induces a path through $D$. Let for every $i\in [\ell]$ the Boolean function $f^i$ be defined as the function accepting exactly the assignments $a$ accepted by $D$ and whose path leads through $a$. Then $f^i$ is a rectangle, since we can freely combine the paths from the source to $v_i$ with those from $v_i$ to the $1$-sink. Moreover, every assignment accepted by $D$ must be accepted by at least one $v_i$ since the $1$-sink is not labeled by $X$ but the source of $D$ is (ignoring trivial cases here in which $f$ does not depend on $X$ where the statement is clear because $f$ itself is a rectangle cover of itself). So we get that 
		\begin{align*}
			f(X,Y) = \bigvee _{i\in [\ell]} f^i(X,Y)
		\end{align*}
	and thus $f$ has a rectangle cover of size $\ell$. Observing that $D$ has at least the nodes $v_1, \ldots, v_\ell$ which are all different and thus at least size $\ell$, completes the proof of the claim.
	\end{claimproof}
	We will show a lower bound on the size of any rectangle cover of the function $f(X,Y) := \bigwedge_{i\in [n]}x_i\lor y_i$, using the so-called fooling set method as follows: consider the set $M$ of models of $f(X,Y)$ of Hamming weight exactly $n$. These models assign for each $i\in [n]$ exactly one of $x_i$ and $y_i$ to $1$. 
	
	\begin{claim}\label{clm:fooling}
		In any rectangle cover of $f$, no two assignments $a,b\in M$ with $a\ne b$ lie in the same rectangle.
	\end{claim}
	\begin{claimproof}
		By way of contradiction, assume that there is a rectangle cover of $f$ such that there are $a,b\in M$ with $a\ne b$ that lie in the same rectangle $f^j$. Let $a_X$ be the restriction of $a$ to $X$ and $a_Y$ that to $Y$. Define $b_X$ and $b_Y$ analogously. Since $a$ and $b$ are not equal, there is an $i\in [n]$ where $a(x_i) \ne b(x_i)$ or $a(y_i)\ne b(y_i)$. Since we have by the choice of $M$ that $a(x_i) = \neg a(y_i)$ and similarly $b(x_i) = \neg b(y_i)$, we actually get that $a(x_i) \ne b(x_i)$ and $a(y_i)\ne b(y_i)$ are both true. It follows that $a(x_i) = b(y_i) \ne a(y_i) = b(x_i)$.
		
		Now assume w.l.o.g.~that $a(x_i) = b(y_i) = 0$. Then for $c=a_X \cup b_Y$ we have $f(c) = 0$. But since $f^j$ is a rectangle, we have that $f^j(c) = f^j_1(a_X)\land f^j(b_Y)= 1$ which is a contradiction. So $a$ and $b$ cannot lie in the same rectangle, as claimed.
	\end{claimproof}
	Note that for every assignment $a_X$ to $X$ there is an extension to $Y$ such that the resulting assignment $a$ is in $M$ (simply set for every $i\in [n]$ the missing value by $a(y_i):= \neg a(x_i)$). Thus, $M$ has size $2^n$. By Claim~\ref{clm:fooling} and Claim~\ref{clm:rectangles}, we get that any rectangle cover of $f$ and thus any BDD for $f$ has size $2^n$ which completes the proof.
\end{proof}

\section{A Polynomial Upper Bound for the Pigeonhole Principle}\label{sct:upper}

We consider SAT-encodings of pigeonhole problems on bipartite graphs $G=(A,B,E)$\footnote{We remark that the same formulas are called \emph{bipartite perfect matching benchmarks} in~\cite{CodelRHB21,CodelRHB21b}, but since the name \emph{perfect matching principle} is used for a related but different class of formulas in proof complexity~\cite{Razborov04}, we follow the notation from~\cite{Ben-SassonW01} here and speak of \emph{pigeonhole formulas} to avoid confusion.}. We assume that $|B|>|A|$, so there is no perfect matching in the graph. In the graphs we consider, we will have $A=[n]$ and $B=[n+1]$. We encode the non-existence of a perfect matching by generalizing the usual direct encoding of the pigeonhole principle: for every edge $ij\in E$, we introduce a variable $p_{i,j}$ which encodes if the edge $ij$ is put into a matching or not. 
For every $i\in A$, we encode by an at-most-one constraint
\begin{align*}
	\text{AMO}_i := \bigwedge_{j,k\in N(i), j\ne k} \bar p_{i,j} \lor \bar p_{i,k}, 
\end{align*}
that at most one vertex from $B$ is matched to $i$. Here, $N(i)$ is the neighborhood of $i$ in $G$, i.e., the set of vertices connected to $i$ by an edge.
For every $j \in B$, we add a clause
\begin{align*}
	\text{ALO}_j := \bigvee_{i\in N(j)} p_{i,j}
\end{align*}
encoding the fact that $j$ must be matched to one of its neighbors in $A$.

The pigeonhole formula for $G$ is then
\begin{align*}
	G\text{-PHP} := \bigwedge_{j\in B} \text{ALO}_j \land \bigwedge_{i\in A} \text{AMO}_i.
\end{align*}

We recover the usual pigeonhole problem formula $\text{PHP}_n$ by considering the complete bipartite graph $K_{n, n+1}=([n], [n+1], [n]\times [n+1])$. 
Conversely, we get $G\text{-PHP}$ from $\text{PHP}_n$ by the restriction that sets the variables $p_{i,j}$ for $ij\notin E$ to $0$.

It is useful to consider the variables $p_{i,j}$ of $\text{PHP}_n$ organized in a matrix where, as usual,~$i$ gives the row index while $j$ gives the column index. Note that with this convention, $\text{ALO}_j$ only has variables in column $j$ while $\text{AMO}_i$ only has variables in row $i$.

We consider two orders on the variables in $\text{PHP}_n$: the \emph{row-wise order}
\begin{align*}
	\pi_r := p_{1,1}, p_{1,2}, \ldots, p_{1,n+1}, p_{2,1}, \ldots p_{n, n+1}
\end{align*}
that we get by reading the variable matrix row by row and the \emph{column-wise order}
\begin{align*}
	\pi_c := p_{1,1}, p_{2,1}, \ldots, p_{n,1}, p_{1,2}, \ldots p_{n, n+1}
\end{align*}
that we get by reading the variable matrix column by column. We consider the same orders for subgraphs $G$ of $K_{n,n+1}$ by simply deleting the variables of edges not in $G$.

\begin{theorem}\label{thm:upperpigeons}
	Bucket elimination in which all BDDs have order $\pi_r$ and the elimination proceeds in order $\pi_c$ refutes $\text{PHP}_n$ in polynomial time.
\end{theorem}
\begin{proof}
	We will polynomially bound the size of all BDDs constructed by the algorithm; since all BDD operations we use can be performed in time polynomial in the BDD size~\cite{Bryant86}, the result then follows directly. In this we tacitly also use the fact that all operations on BDDs we use the constructed BDDs can be assumed to be a minimal size for the variable order due to canonicity of BDDs. We first analyze the BDDs that result from the respective quantification steps (Line~\ref{ln:project} in Algorithm~\ref{alg}). 
	We denote by $F'_{i,j}$ the function computed by the BDD in which we quantify $p_{i,j}$. By $F_{i,j}$ we denote the CNF formula that is the conjunction of all clauses that have been conjoined before this elimination step.
	Observe that we get $F_{i, j}'$ from $F_{i,j}$ by quantifying all variables up to $p_{i,j}$ in $\pi_c$. Moreover, if $p_{i,j}$ is before $p_{i', j'}$ in $\pi_c$, then the clauses in $F_{i,j}$ are a subset of those in $F_{i', j'}$.

	$F_{i,j}$ consists of all clauses of $\text{PHP}_n$ that have a variable up to $p_{i,j}$ in the order $\pi_c$, so
	\begin{enumerate}
		\item the clauses $\text{ALO}_k$ for all $k\le j$,\label{pigeon}
		\item the clauses $\bar p_{i', k} \lor \bar p_{i',\ell}$, for $i' \in [n]$ and $1\le k< \ell \le n+1$, $k < j$, and \label{holecomplete}
		\item the clauses $\bar p_{i',j} \lor \bar p_{i', k}$ for $i' \in [i]$ and $k \in [n+1], k\ne j$.\label{holepartial}
	\end{enumerate}
	Remember that we get $F_{i,j}'$ from $F_{i,j}$ by quantifying the variables up to $p_{i,j}$ in the order $\pi_c$. We will show that $F_{i,j}'$ can be encoded by a small BDD. We first consider the case $j=1$.
	
	\begin{claim}
	An assignment $a'$ satisfies $F_{i,1}'$ if and only if 
	\begin{itemize}
	\item $a'$ sets one of the $p_{i', 1}$ with $i'\in [i+1,n]$ to $1$, or 
	\item there is an $i^*\in [i]$ such that all $p_{i^*,j}$ with $j\in [2,n+1]$ take the value $0$ in $a'$.
	\end{itemize}
	\end{claim}
\begin{claimproof}
 Assume first that there is a $p_{i', 1}$ with $i'\in [i+1,n]$ set to $1$ by $a'$. We extend $a'$ to an assignment $a$ of $F_{i,1}$ by setting all quantified variables $p_{k,1}$ for $k\in [i]$ to $0$. Then $\text{ALO}_1$ is satisfied by $p_{i', 1}$ and the $p_{k, 1}$ for $k\in [i]$ satisfy the clauses in \ref{holepartial}. Since $F_{i,1}$ does not have clauses from \ref{holecomplete}, $F_{i,1}$ is satisfied by $a$ and thus $a'$ satisfies $F_{i,1}'$. If there is an $i^*\in [i]$ such that all $p_{i^*,j}$ with $j\in [2,n+1]$ take the value $0$, then set $a(p_{i^*, 1}):= 1$ and $a(p_{k, 1}):= 0$ for all other $k\in [i], k\ne i^*$. As before, all clauses of $F_{i,1}$ are satisfied and thus $F_{i,1}'$ is satisfied by $a'$. 
 
 For the other direction, assume that $a'$ satisfies $F_{i,1}'$ and that $a$ is an extension of $a'$ that satisfies $F_{i,1}$. If there is an $i^*\in [i]$ such that all $p_{i^*,j}$ with $j\in [2,n+1]$ take the value $0$, then there is nothing to show. So assume that for every $k\in [i]$ there is a $j'\in [2, n+1]$ such that $a'(p_{k,j'})= 1$. Since $F_{i,1}$ contains the clause $\bar p_{k,1} \lor \bar p_{k, j'}$, we have $a(p_{k,1})= 0$. Since this is true for all $k\in [i]$ and $a$ satisfies $\text{ALO}_1$, there must be $i'\in [i+1, n]$ which is set to $1$ by $a'$ which completes the proof of the claim. 
\end{claimproof}
	It follows that $F_{i,1}$ can be expressed as a small BDD with variable order $\pi_r$: check for every fixed $i'\in [i]$ if the value of all $p_{i,j}$ is $0$. Since, for every $i\in [n]$, these variables are consecutive in $\pi_r$, this can be easily done by a BDD that is essentially a path. We then glue these BDDs in increasing order of $i'$ in the obvious way and check for $i'\in [i+1,n]$ if $p_{i',1}$ takes value $1$ to get a BDD for $F_{i,1}$ of size $O(n^2)$, since we have to consider $O(n^2)$ variables.

	We now consider the case $j>2$. In that case, $F_{i,j}$ contains all variables of $\text{PHP}_n$. Note that if $j=n+1$, then $F_{i,j}= \text{PHP}_n$ and thus the formula $F_{i,j}'$ is unsatisfiable and has a constant size encoding as a BDD. So assume in the remainder that $j\le n$. Consider an assignment $a'$ to $F_{i,j}'$. Let $I_{a'}$ be the set of indices $i'\in [n]$ such that for all $j'$ for which $p_{i', j'}$ appears in $F_{i,j}'$ we have $a'(p_{i',j'})= 0$.
	\begin{claim}\label{clm:generalpigeon}
		$a'$ satisfies $F_{i,j}'$ if and only if
		\begin{enumerate}
			\item $|I_{a'}| \ge j$ and there is an index $i'\in I_{a'}\cap [i]$, or\label{en:case1}
			\item $|I_{a'}| \ge  j-1$ and there is an $i^*\in [i+1, n]\setminus I_{a'}$ such that $a'(p_{i^*, j})=1$.\label{en:case2}
		\end{enumerate}
	\end{claim}
	\begin{claimproof}
		Let first Case \ref{en:case1} be true. Then we can construct a injective function $f:[j]\rightarrow I_{a'}$ with $f(j)\in I_{a'}\cap [i]$. We construct an extension $a$ of $a'$ to all variables of $F_{i,j}$ as follows: for $j'\in [j]$ we set $a(p_{f(j'), j'}):= 1$ and set all other variables to $0$. Then for $j'\in [j]$, the clause $\text{ALO}_{j'}$ is satisfied by $p_{f(j'), j'}$. Let $V_{i,j}$ the variables of $a$ not assigned in $a'$. For every $i' \notin I_{a'}$, the variables $p_{i', k}\in V_{i,j}$ are assigned to $0$, so all clauses of the form $\bar p_{i', k} \lor \bar p_{i',\ell}$ in $F_{i,j}$ are satisfied. If $i'\in I_{a'}$, then, since $f$ is injective, at most one variable $p_{i', k}$ with $k\in [n+1]$ is assigned to $1$, so $a$ satisfies $\text{AMO}_i$ and thus in particular all clauses of the form $\bar p_{i', k} \lor \bar p_{i',\ell}$. Thus, $a$ satisfies $F_{i,j}$ and $a'$ satisfies $F_{i,j}'$.
		
		Now assume Case~\ref{en:case2} is true. Construct a injective function $f:[j-1]\rightarrow I_{a'}$. We again construct an extension $a$ of $a'$:  for $j'\in [j-1]$ we set $a(p_{f(j'), j'}):= 1$ and set all other variables in $V_{i,j}$ to $0$. We show that all clauses of $F_{i,j}$ are satisfied by $a$. First, $\text{ALO}_j$ is satisfied by $p_{i^*, j}$. For $j'\in [j-1]$, the clause $\text{ALO}_j$ is satisfied by $p_{f(j'), j'}$. For the binary clauses, we reason exactly as in the previous case. It follows that $a'$ satisfies $F_{i,j}$.
		
		For the other direction, assume that $a'$ satisfies $F_{i,j}'$ and let $a$ be an extension of $a'$ that satisfies $F_{i,j}$. Since $a$ must in particular satisfy the clauses $\text{ALO}_{j'}$ for $j'\in [j-1]$, we can choose, for every $j'\in [j-1]$, an index $f(j')\in [n]$ such that $a(p_{f(j'), j'})= 1$. All variables in the clauses $\text{ALO}_{k}$ for $k\in [j-1]$ are in $V_{i,j}$, so all binary clauses $p_{i', k} \lor \bar p_{i',\ell}$ with $k\in [j-1]$ appear in $F_{i,j}$. So in particular, there cannot be $k\in [j-1], \ell\in [n+1]$ such that $a(p_{f(k), k}) = a(p_{ f(k), \ell})=1$. It follows that $f$ is injective and for every $k\in [j-1]$ we have that $f(k)\in I_{a'}$. It follows that $|I_{a'}| \ge j-1$.
		
		Now assume that Case~\ref{en:case1} is false. Say first that $|I_{a'}|\ngeq j$, so $|I_{a'}|=j-1$. Then $f$ is a bijection. The clause $\text{ALO}_j$ is satisfied by $a$, so there must be $i^*\in [n]$ such that $a(p_{i^*, j})=1$. We claim that $i^*$ cannot be in $I_{a'}$. By way of contradiction, assume this were wrong. Then, because $f$ is a bijection, there is $j'\in [j-1]$ with $f(j')=i^*$. By construction of $f$, we have $a(p_{f(j'), j'})= a(p_{i^*, j'})=1$. Then, because $F_{i,j}$ contains the clause $\bar p_{i^*, j'} \lor \bar p_{i^*, j}$, the assignment $a$ does not satisfy $F_{i,j}$ which is a contradiction. So $i^*\notin I_{a'}$.  Moreover, $i^*> i$ since otherwise all binary clauses $\bar p_{i^*, j} \lor \bar p_{i^*,\ell}$ would be in $F_{i,j}$ and thus $i^*$ would be in $I_{a'}$. So in this case we have that Case~\ref{en:case2} is true.
		
		If there is no index $i'\in I_{a'}\cap [i]$, then we claim that $\text{ALO}_j$ is satisfied by a variable $p_{i^*, j}$ for $i^*> i$: reasoning with the binary clauses similarly to before, whenever $a(p_{i', j})=1$ for some $i'\in [i]$, then $i'\in I_{a'}$. So none of the $p_{i', j}$ with $i'\in [i]$ satisfy $\text{ALO}_j$ and it is satisfied by some $p_{i^*, j}$ which appears in $F_{i,j}'$. Then $i^*\notin I_{a'}$ due to $a'(p_{i^*, j}) = 1$, so Case~\ref{en:case2} is true.
	\end{claimproof}	
	
	With Claim~\ref{clm:generalpigeon}, we can bound the size of the BDD encoding $F_{i,j}'$: since for every $i'\in [n]$ the variables $p_{i',j'}$ are consecutive in the order $\pi_r$, we can check if $i'\in I_{a'}$ by a BDD of constant width. By making $j$ parallel copies of this BDD for every $i'$, we can compute the size of $I_{a'}$ cutting off at $j$ in width $O(j)$. We can also check if there is an index $i'\in I_{a'}\cap [i]$ or $i^*\in [i+1, n]\setminus I_{a'}$ such that $a'(p_{i^*, j})=1$ with only a constant additional factor. Since $F_{i,j}'$ has $O(n(n-j))$ variables, the overall size of the BDD computing $F_{i,j}'$ is $O(j(n-j)n)= O(n^3)$.
	
	It remains to bound the size of BDDs we get from the conjoin-steps between quantification steps. So consider a conjoin step before quantifying $p_{i,j}$ but after the potential previous quantification. Call the resulting BDD $D$ and let $D'$ be the BDD we got from the previous quantification (if there is no previous quantification, set $D'$ to the constant $1$ BDD).
	
	\begin{claim}\label{clm:apply}
		The size of $D$ is $O(n^3)$.
	\end{claim}
	\begin{claimproof}
		We first claim that when we start conjoining the BDDs in the bucket of variable $p_{i,j}$, the only BDD that does not encode a clause is $D'$. This is because after every quantification step the result contains the next variable in the order $\pi_c$. Thus, since we conjoin only BDDs that contain the variable $p_{i,j}$, the BDDs involved in these steps are $D'$ and potentially BDD representations of $\text{ALO}_j$ and clauses $\bar p_{i,j} \lor \bar p_{i,k}$ for $k>j$. First assume that the conjunction only involves clauses $\bar p_{i,j} \lor \bar p_{i,k}$ for $k>j$. We claim that the result then has size $O(d)$ where~$d$ is the number of conjuncts involved. To see this, observe that if $p_{i,j}$ takes value $0$, then all clauses in the conjunction are true, so the conjunction evaluates to true as well so we can directly go to the $1$-sink. If $p_{i,j}$ takes value $1$, then we have to verify if all other $p_{i,k}$ involved in the conjunction are $0$ which can be done by a path of length $d$ because $p_{i,j}$ is the first variable in $\pi_r$. 
		
		If the conjunction also involves $\text{ALO}_j$, then if $p_{i,j}$ takes value $1$, we proceed as before since $\text{ALO}_j$ is satisfied already. For the case where $p_{i,j}$ takes value $0$, we have to check all other variables in $\text{ALO}_j$ on a path. $\text{ALO}_j$ is only conjoined if $i=1$, so in that case $p_{i,j}$ is again the first variable to consider, so this procedure can be done following the order $\pi_r$. Overall, the conjunction in this case has size $O(n)$. Note also that in all cases discussed so far, we can also represent the conjunction by a BDD of constant width.
		
		It remains to consider the case in which $D'$ is involved in the conjunctions. We can then see the conjunction as one of several clauses, as discussed above, and $D'$. As shown above,~$D'$ has width $O(j)= O(n)$, so this is also true for the result $D$ of conjoining some of the clauses, since the latter contribute only constant width. So $D$ has a BDD of size $O(n^3)$.
	\end{claimproof}
	
	We have shown that all BDDs that we ever construct in the refutation have size at most $O(n^3)$. We make $O(n^3)$ conjoin operations and $O(n^2)$ quantifications and all BDD-operations can be performed in time polynomial in the input, so the overall runtime is polynomial.
\end{proof}

\section{No Closure Under Restrictions}

We next show that Theorem~\ref{thm:upperpigeons} is not true for restrictions of $\text{PHP}_n$.
To this end, we consider the graph $G= ([2n], [2n+1], E)$ where the edge set $E$ is defined by 
\begin{align*}
	E= \{(j, j), (n+j, j), (j, n+1+j), (n+j, n+1+j), (j, n+1), (n+j, n+1), \mid j\in [n]\}.
\end{align*}
\begin{theorem}\label{thm:uppermatchings}
	Bucket elimination in which all BDDs have order $\pi_r$ and the elimination proceeds in order $\pi_c$ refutes $G\text{-PHP}$ in time $\Omega(2^n)$.
\end{theorem}
\begin{proof}
	We will show that bucket elimination constructs an exponential size BDD in its run. To this end, we first give all the clauses of $G\text{-PHP}$ (with the constraint names below):
	\begin{align*}
		&\bigwedge_{j\in [n]} \underbrace{\big((\bar p_{j,j}\lor \bar p_{j,n+1+j}) \land (\bar p_{j,j}\lor \bar p_{j,n+1}) \land (\bar p_{j,n+1}\lor \bar p_{j,n+1+j})\big)}_{\text{AMO}_j} 
		\\
		\land &\bigwedge_{j\in [n]} \underbrace{\big((\bar p_{n+j,j}\lor \bar p_{n+j,n+1+j}) \land (\bar p_{n+j,j}\lor \bar p_{n+j,n+1}) \land (\bar p_{n+j,n+1}\lor \bar p_{n+j,n+1+j})\big)}_{\text{AMO}_{n+j}}
		\\
		\land &\bigwedge_{j\in [n]} \big(\underbrace{(p_{j,j} \lor p_{n+j, j})}_{\text{ALO}_j} \land \underbrace{(p_{j, n+1+j} \lor p_{n+j, n+1+j})}_{\text{ALO}_{n+1+j}}\big) \land \underbrace{\bigvee_{j\in [n]} p_{j, n+1} \lor p_{n+j, n+1}}_{\text{ALO}_{n+1}}
	\end{align*}
We consider the step directly before the quantification of $p_{2n, n+1}$, so after conjoining the contents of $B_{p_{2n, n+1}}$ to a BDD $D$ in Line~\ref{ln:conjoin} in Algorithm~\ref{alg}. We claim that $D$ has exponential size. To this end, first observe that, at the time of the construction of $D$, all clauses have been joined except $\text{ALO}_{n+1+j}$ which contain no variables $p_{i,j'}$ with $j'\in [n+1]$. We claim that all these clauses have contributed to $D$. Indeed, whenever eliminating $p_{i,j}$ with $j\in [n]$, the result contains the variable $p_{i, n+1}$ and will thus be put into the bucket $B_{p_{i,n+1}}$ eventually. Then the clause $\text{ALO}_{n+1}$ makes sure that all these BDDs are (after some more conjoining and quantification) contributing to $D$. So we get $D$ by conjoining all clauses except the $\text{ALO}_{n+1+j}$ and eliminating all variables up to $p_{2n-1, n+1}$.

Let $F$ be the function we get from $D$ by fixing $p_{n, 2n+1}, p_{2n, 2n+1}$ to $0$ and $p_{2n, n+1}$ to $1$. Let $F'$ be the corresponding conjunction of clauses. Then $\text{ALO}_{n+1}$ is satisfied and the remaining literals $\bar p_{i, n+1}$ are all pure in $F'$. Thus, by pure variable elimination, an assignment $a$ to the variables $p_{j, n+1+j}, p_{n+j, n+1+j}$ for $j\in [n-1]$, which are the variables of $F$, can be extended to a satisfying assignment of $F'$ if and only if it can be extended to a satisfying assignment of 
	\begin{align*}
	\bar p_{2n, n} \land (p_{n,n}\lor p_{2n, n})\land \bigwedge_{j\in [n-1]} (\bar p_{j,j} \lor \bar p_{j,n+1+j}) \land (\bar p_{n+j, j}\lor \bar p_{n+j, n+1+j}) \land (p_{j, j}\lor p_{n+j,j}).
\end{align*}
Eliminating $p_{j,j}, p_{n+j,j}$ for $j\in [n]$, we see that $F$ is equivalent to 
	\begin{align*}
	\bigwedge_{j\in [n-1]}(\bar p_{j,n+1+j} \lor \bar p_{n+j, n+1+j}).
\end{align*}
When representing $F$ in a BDD with row-wise variable order, all $p_{j, n+1+j}$ are before all $p_{n+j, n+1+j}$, so we are, up to renaming literals which does not change the size of a BDD, in the situation of Lemma~\ref{lem:lowerscov}. We get that any BDD for $F$ with the order $\pi_r$ has size at least $2^{n-1}$ and, since fixing variables does not increase the size of BDD-representations, we get the same lower bound for $D$.
\end{proof}

\section{Lower Bounds for Single Orders}\label{sct:lower}

We will now analyze bucket elimination in which the elimination order is also the order in which variables appear in the BDDs. This is the behavior of the implementations of~\cite{PanV05,JussilaSB06}; the implementation of~\cite{BryantH21} allows the use of two orders in the current version.
In the variant with one order, which we call \emph{single-order bucket elimination}, it is always the variable in the source of the BDDs that is eliminated, which makes the algorithm simpler. 
We show here that neither of the two orders introduced in Section~\ref{sct:upper} leads to polynomial runtime behavior on its own, suggesting that it is the combination of both that is required for efficiency.

\begin{lemma}\label{lem:lowercol}
Single-order bucket elimination for $\text{PHP}_{n}$ with order $\pi_c$ constructs an intermedi\-ate BDD of size $2^{n}$.
\end{lemma}
\begin{proof}
Consider the situation after we have eliminated the variables $p_{1,1}, p_{2,1}, \ldots, p_{n,1}$. As analyzed in the proof of Theorem~\ref{thm:upperpigeons}, after eliminating the last of these variables, we have constructed a BDD for the function $F_{n,1}'$ that is satisfied by an assignment if and only if there is an $i$ such that all $p_{i,j}$ with $j\in [2,n+1]$ take the value $0$. We claim that the BDD-representation of $F_{n,1}'$ has exponential size.

To show this, we consider the restriction $F$ of $F_{n,1}'$ that we get by fixing all variables $p_{i,j}$ for $j>3$ to $0$. Thus, $F$ has the variables $p_{1,2}, p_{1,3}, p_{2,2}, p_{2,3}, \ldots, p_{n, 2}, p_{n,3}$. We rename for all $i\in [n]$ the variables $p_{i,2}$ to $x_i$ and $p_{i,3}$ to $y_i$. The resulting function $F'$ evaluates to $1$ if and only if there is an $i$ such that $x_i$ and $y_i$ take the value $0$. Then the negation $\bar F$ of $F$ 
is given by
$\bar F = \bigwedge_{i\in [n]} x_i\lor y_i$.
Moreover, the variable order of the BDD we have to consider has all $x_i$ before any $y_i$, and thus, by Lemma~\ref{lem:lowerscov}, any BDD for $\bar F$ has size at least $2^n$. Since BDDs allow negation and restrictions without size increase, this shows the lower bound for $F_{n, 1}$ and thus the claim.
\end{proof}

\begin{lemma}
	Single-order bucket elimination for $\text{PHP}_{n}$ with order $\pi_r$ constructs an intermedi\-ate BDD of size $2^{n}$.
\end{lemma}
\begin{proof}
Consider the BDD $D$ we construct after eliminating the first row. The clauses that contribute to this function are all clauses of $\text{AMO}_1$ as well as all $\text{ALO}_j$ for all $j\in [n+1]$. When quantifying away the $p_{1,j}$ for $j\in [n]$, we get a function $F$ that is satisfied by an assignment $a$ if there is a $j^*\in [n+1]$ such that for all $j\in [n+1]\setminus \{j^*\}$ there is a $p_{i,j}$ set to $1$ by $a$; this is because all $\text{ALO}_j$ have to be satisfied and at most one of them can be satisfied by $p_{1,j}$ due to $\text{AMO}_1$. $F$ has to be represented by a BDD in bucket elimination and we will show that this requires exponential size.
To see this, fix all variables $p_{i,j}$ for $i>3$ to $0$ and fix $p_{2,n+1}$ and $p_{3,n+1}$ to $0$. The resulting function is 
$F' = \bigwedge_{j\in [n]} p_{2,j}\lor p_{3,j}$.
In the BDD-representation, all $p_{2,j}$ come before any $p_{3, j}$, so, 
up to renaming the variables $p_{2,j}$ to $x_j$ and $p_{3,j}$ to $y_j$, we are in the situation of  Lemma~\ref{lem:lowerscov} and, observing that fixing variables does not increase the size of a BDD, the lower bound follows from there.
%
%
\end{proof}

\section{Conclusion}

We have shown that bucket elimination based SAT-solving using BDDs can efficiently solve pigeonhole principle formulas, theoretically confirming prior experimental work from~\cite{CodelRHB21,CodelRHB21b,BryantH21}, which worked with a slightly different encoding. We have also seen that this result is not stable under restrictions, showing that, at least for the same orders, there are formulas we get by restriction of the pigeonhole principle that take exponential time to solve. We have also seen that the common single-order variant of bucket elimination~\cite{PanV05,JussilaSB06,BryantH21} has exponential runtime for the two variable orders that in combination allow efficient solving. 

For practical SAT-solving with BDD-based solvers, our results are mixed news: while we confirm that these solvers are in principle powerful in the sense that they can efficiently solve instances that are out of reach for resolution and thus CDCL-solvers, our additional results suggest that in general it might be hard to come up with the right two variable orders for the instances at hand, in particular since orders good for one type of formulas are bad for very related formulas. So it is not clear how useful BDD-based bucket elimination will be beyond very restricted formula classes.

We close the paper with some questions. First, it
would be interesting to understand if for every bipartite graph $G$ the formula $G$-PHP can be refuted efficiently by choosing orders adapted to the problem or if there are graphs for which bucket elimination is slow for all order choices. In particular, one might also consider some of the many different variants of the pigeonhole principle or mutilated chessboard formulas which have been considered extensively in the literature as benchmarks for solvers but also in theoretical work, see e.g.~the overview in~\cite{Razborov04}.

Finally, it is not clear if single-order bucket elimination can solve $\text{PHP}_n$ efficiently for some order. The experimental work in~\cite{Bryant86,CodelRHB21b} does not show any such order, and our own search in this direction has shown only lower bounds that are variants of those presented in Section~\ref{sct:lower}. It is thus natural to conjecture that in fact single-order variable bucket elimination cannot solve $\text{PHP}_n$ efficiently. Note that proving this would in particular show that two orders make the approach strictly more powerful.

\bibliographystyle{plain}
\bibliography{buckets}

\end{document}